\documentclass[letterpaper, 10 pt, conference]{ieeeconf}
\IEEEoverridecommandlockouts{}   
\usepackage{cite}
\usepackage{amsmath,amssymb,amsfonts}
\usepackage{algorithmic}
\usepackage{graphicx}
\usepackage{textcomp}
\usepackage{dsfont}
\usepackage{color}
\usepackage{epstopdf}        
\usepackage{enumerate}
\usepackage{soul}
\usepackage{lscape}
\usepackage{multicol}
\usepackage{multirow}
\usepackage{calligra}
\usepackage{booktabs}
\usepackage{mathtools}
\usepackage{framed} 
\usepackage{picins}
\usepackage{empheq}
\usepackage{afterpage}
\usepackage{graphics} 
\usepackage{epsfig} 
\usepackage{times} 
\usepackage{amsmath} 
\usepackage{amssymb}  
\usepackage{amsfonts}  
\usepackage{subfig}
\usepackage{epstopdf}
\usepackage{enumerate}
\usepackage{graphicx} 
\usepackage{xcolor}
\usepackage{comment}

\newtheorem{thm}{Theorem}

\newtheorem{lemma}{Lemma}

\newtheorem{definition}{Definition}
\newtheorem{assumption}{Assumption}
\newtheorem{remark}{Remark}

\newcommand\tx{\tilde{x}}
\newcommand\tu{\tilde{u}}
\newcommand\bj{\mathbf{J}}
\newcommand\re{\mathbb{R}}

\newcommand\cc{\mathcal{C}}

\newcommand\cf{\mathcal{F}}
\newcommand\tcf{\tilde{\mathcal{F}}}
\newcommand\ck{\mathcal{K}}

\newcommand\cki{\mathcal{K}_\infty^{-1}}
\newcommand\ckf{\mathcal{K}_\infty^{\text{FxT}}}

\DeclarePairedDelimiter\sg{\lfloor}{\rceil}
\newcommand\ckp{\mathcal{K}_\infty^+}
\newcommand\klfx{\mathcal{KL}^{\text{FXT}}}
\newcommand\lmax{\overline{\lambda}}
\newcommand\lmin{\underline{\lambda}}
\newcommand\smax{\overline{\mu}}

\newcommand*{\QEDB}{\hfill\ensuremath{\square}}%

\def\BibTeX{{\rm B\kern-.05em{\sc i\kern-.025em b}\kern-.08em
    T\kern-.1667em\lower.7ex\hbox{E}\kern-.125emX}}
\begin{document}
\title{A Lyapunov-Based Small-Gain Theorem for Fixed-Time Stability}

\author{Michael Tang, Miroslav Krstic, Jorge Poveda
\thanks{M. Tang, J. I. Poveda are with the Dep. of Electrical and Computer Engineering, University of California, San Diego, La Jolla, CA, USA.}
\thanks{ M. Krsti\'c is with the Dep. of Mechanical and Aerospace Engineering, University of California, San Diego, La Jolla, CA, USA.}
\thanks{This work was supported in part by NSF grants ECCS CAREER 2305756, CMMI 2228791, and AFOSR FA9550-22-1-0211.}
\vspace{-0.6cm}
}
\maketitle

\begin{abstract}
This paper introduces a novel \emph{Lyapunov-based small-gain methodology} for establishing fixed-time stability (FxTS) guarantees in interconnected dynamical systems. Specifically, we consider interconnections in which each subsystem admits an individual fixed-time input-to-state stability (ISS) Lyapunov function that certifies FxT-ISS. We then show that if a nonlinear small-gain condition is satisfied, then the entire interconnected system is FxTS. Our results are analogous to existing Lyapunov-based small-gain theorems developed for asymptotic and finite-time stability, thereby filling an important gap in the stability analysis of interconnected dynamical systems. The proposed theoretical tools are further illustrated through analytical and numerical examples, including the first result on fixed-time feedback optimization of dynamical systems without time-scale separation between the plant and the controller.
\end{abstract}
\section{Introduction}
To address the growing demand for fast control and learning algorithms with greater accuracy and stronger robustness guarantees, significant research efforts have focused on developing tools that enable improved convergence properties. Finite-time stability (FTS) \cite{bhatfinite} and control \cite{zpjfinite} have attracted attention from the nonlinear controls community due to improved convergence and disturbance rejection properties. However, a drawback of finite-time stabilization is that the time by which the system is guaranteed to reach its equilibrium (i.e., the settling time) may grow unbounded with respect to the system's initialization. This motivates the notion of fixed-time stability (FxTS), which distinguishes itself from FTS by further guaranteeing a uniform upper bound on the settling time. Following the introduction of Lyapunov conditions in \cite{6104367} to verify FxTS for continuous-time autonomous systems, FxTS has received significant attention due to its ability to address challenges in a variety of engineering domains, such as optimization \cite{10885949}, learning \cite{9760031}, neural networks \cite{CHEN2020412}, etc. Despite the numerous advancements in the theory of FxTS, there are limited tools available for the study of FxTS in large scale systems. For example, current methods require the systems to satisfy restrictive structural requirements \cite{9714166}, certain interconnection conditions under a time-scale separation \cite{11239427}, or trajectory-based small gain conditions \cite{10886634}. 

On the other hand, the small-gain theorem is a fundamental tool for the stability analysis of interconnected systems. Originally developed for systems with linear gains \cite{1098316}, the small-gain theorem has been extended to those with nonlinear gains \cite{jiang1994small} after the development of input-to-state stability (ISS) theory \cite{sontag2008input}. In the past few decades, small-gain theory has been extended to hybrid systems \cite{bao2018ios,liberzon2012small}, stochastic systems \cite{dragan1997small}, discrete-time systems \cite{jiang2004nonlinear,zhongping2008nonlinear}, and infinite dimensional systems \cite{mironchenko2021nonlinear, mironchenko2021small}, to name a few.
Moreover, to harness the power, applicability, and constructiveness of Lyapunov-based methods, considerable work has also been done to reformulate nonlinear small-gain theory using Lyapunov functions \cite{jiang1996lyapunov,liberzon2014lyapunov,kawan2020lyapunov}. While the literature on nonlinear small-gain theory is indeed vast, extensions to systems with accelerated convergence are relatively recent and, for the most part, underdeveloped. For example, trajectory-based small-gain conditions were derived in \cite{zpjfinite} for interconnected finite-time input-to-state stable systems, and a Lyapunov-based reformulation for networks was proposed in \cite{sgfinite}. Following the development of the concept of fixed-time input-to-state stability (FxT ISS) \cite{LOPEZRAMIREZ2020104775}, a trajectory-based small-gain theorem was developed for interconnected FxT ISS systems in \cite{10886634}. However, as acknowledged by the authors, the results are quite restrictive in that they contain additional conditions that depend on the settling-time function, its inverse, and several auxiliary functions that are used to upper bound the trajectories of the subsystems.

The main contribution of this paper is the introduction of a novel Lyapunov-based small-gain theorem to establish FxTS for interconnected FxT ISS subsystems. Specifically, we show that if each of the subsystems admit a FxT ISS Lyapunov function, and the gain functions ---under a mild structural requirement--- satisfy the nonlinear small-gain condition, then the overall interconnected system is FxTS. In this sense, our main result fills an important gap in the stability of nonlinear systems and mirrors similar Lyapunov-based small gain theorems developed in the literature for the study of asymptotic stability and finite-time stability. However, as described in the paper, our results do not follow as simple extensions of previous results, but rather require the derivation of several new technical lemmas and results, which are presented in the Appendix for the sake of completeness.

After introducing the main analytical results, we present two examples to illustrate our main results. In the first example we consider the interconnection of two systems with with homogeneous coupling terms. In the second example we consider the problem of fixed-time feedback optimization in dynamical systems \cite{9075378,8673636,bianchin2022online} without timescale separation, a problem that had not been addressed before, and establish FxTS using our small-gain theorem.

The rest of this paper is organized as follows: Section \ref{sec_preliminaries} presents some mathematical preliminaries. Section \ref{sec_main} presents the main results. Section \ref{sec_fixedtime} presents the application to feedback optimization, and Section \ref{sec_conclusions} presents the conclusions. 

\section{Preliminaries}
\label{sec_preliminaries}
\subsection{Notation}
We use $\re_{\ge 0}$ to denote the set of nonnegative real numbers. A continuous function $\alpha:\re_{\ge 0}\to\re_{\ge 0}$ is said to be of class $\ck$, denoted $\alpha\in\ck$, if $\alpha(0)=0$ and $\alpha$ is strictly increasing. We use $\ck_\infty$ to denote the set of $\alpha\in\ck$ that satisfy $\lim_{s\to\infty}\alpha(s)=\infty$. Given $\alpha\in\ck_\infty$, if there exists $n\in\mathbb{N}$ and constants $c_i>0$, $p_i>0$ for $i=1,...,n$ such that $\alpha(s)=\sum_{i=1}^n c_i s^{p_i}$ for $s\ge 0$, we say $\alpha$ is of class $\ckp$, denoted $\alpha\in\ckp$. Moreover, if $n=2$  with $p_1\in (0,1)$ and $p_2>1$, we say $\alpha\in\ckf$. Given $\alpha\in\ck_\infty$, we use $\alpha^{-1}$ to denote its inverse. We define $\cki$ to be the class of functions $\alpha\in\ck_\infty$ that satisfy $\alpha^{-1}\in\ckp$. A continuous function $\beta:\re_{\ge 0}\times\re_{\ge 0}\to\re_{\ge 0}$ is said to be of class $\klfx$, denoted $\beta\in\klfx$, if $\beta(\cdot, 0)\in\mathcal{K}$ and for each fixed $r\ge 0$, $\beta(r, \cdot)$ is continuous, non-increasing and there exists a function $T:\mathbb{R}_{\ge 0}\to\mathbb{R}_{\ge 0}$ such that $\beta(r, t)=0$ for all $t\ge T(r)$. The mapping $T$ is called the \emph{settling time function}. Given a measurable function $u:\mathbb{R}_{\ge 0}\to\mathbb{R}^m$ we denote $|u|_\infty=\text{ess}\sup_{t\ge 0}|u(t)|$, where $|\cdot|$ represents the Euclidean norm. We use $\mathcal{L}_\infty^m$ to denote the set of measurable functions $u:\mathbb{R}_{\ge 0}\to\mathbb{R}^m$ satisfying $|u|_\infty<\infty$. For a function $f:(a,b)\to\re$, the upper right Dini derivative at $t\in (a,b)$ is defined as $D^+f(t)=\limsup_{h\to 0^+}\frac{f(t+h)-f(t)}{h}$. Given a differentiable function $f:\mathbb{R}^n\to\mathbb{R}^m$, we use $\bj_f(x)\in\mathbb{R}^{m\times n}$ to denote the Jacobian of $f$ evaluated at $x\in\mathbb{R}^n$. If $m=1$, we use $\nabla f(x)=\bj_f(x)^\top$. If $\bj_f(x)$ is continuous, we say $f$ is $\mathcal{C}^1$. We use $I_n\in\re^{n\times n}$ to denote the $n$-by-$n$ identity matrix. If $Q\in\re^{n\times n}$ is symmetric, we use $\lmax(Q)$ and $\lmin(Q)$ to denote its largest and smallest eigenvalue, respectively. If $Q\in\re^{n\times m}$, we use $\smax(Q)$ to denote its largest singular value. We present a simple lemma that will be instrumental for our results
\begin{lemma}\label{lem_sandw}
    Given $\underline{p}\le p\le \overline{p}$, the following holds
    \begin{equation*}
        x^p\le x^{\underline{p}}+x^{\overline{p}},
    \end{equation*}
    for all $x\ge 0$.\QEDB
\end{lemma}
\begin{proof}
    For $x\in [0,1]$ we have $x^p\le x^{\underline{p}}$, and for $x\ge 1$ we have $x^p\le x^{\overline{p}}$. We combine both cases to establish the result.
\end{proof}
\subsection{Fixed-Time stability}
Consider the following class of autonomous systems
\begin{equation}\label{sys}
    \dot{x}=f(x),\quad x(0)=x_0,
\end{equation}
where $x\in\re^n$ is the state and $f:\re^n\to\re^n$ is a continuous vector field with $f(0)=0$. We will state some definitions that are inspired by the results of \cite{6104367}.
\begin{definition}
    The origin of \eqref{sys} is said to be \emph{uniformly globally fixed-time stable (FxTS)} if there exists $\beta\in\klfx$ such that for each $x_0\in\re^n$, every solution of \eqref{sys} exists for $t\ge 0$ and satisfies $|x(t)|\le \beta(|x_0|, t)$
    for all $t\ge 0$,
    where the settling time function $T$ of $\beta$ is uniformly bounded with $T(0)=0$.\QEDB
\end{definition}
\vspace{0.1cm}
\begin{definition}\label{def_fxts_lf}
    A continuous function $V:\re^n\to\re_{\ge 0}$ is said to be a FxTS Lyapunov function for \eqref{sys} if:
    \begin{enumerate}
        \item There exists $\underline{\alpha}, \overline{\alpha}\in\ck_\infty$ such that
        \begin{equation}\label{sandw}
            \underline{\alpha}(|x|)\le V(x)\le \overline{\alpha}(|x|),
        \end{equation}
        for all $x\in\re^n$.
        \item There exists $\Psi\in\ckf$ such that the following holds
    \begin{equation}\label{fxts_lyapunov_bd}
       D^+ V(x(t))\le -\Psi(V(x(t))),
    \end{equation}
    along the trajectories of the system \eqref{sys} for all $t\ge 0$.\QEDB
    \end{enumerate}
\end{definition}
\vspace{0.1cm}
It was shown in \cite{6104367} that the origin of \eqref{sys} is FxTS if it admits a FxTS Lyapunov function. Moreover, if $\Psi$ in \eqref{fxts_lyapunov_bd} is given by $\Psi(s)=as^p+bs^q$ with $a,b>0, p\in (0,1), q>1$, then the settling time of $\beta$ satisfies the following bound for all $x_0\in \re^n$:
\begin{equation*}
    T(x_0)\le \frac{1}{a(1-p)}+\frac{1}{b(q-1)}.
\end{equation*}
\subsection{Fixed-time input-to-state stability}
Consider the following class of dynamical systems that also depend on an input:
\begin{equation}\label{sysu}
    \dot{x}=f(x,u),\quad x(0)=x_0,
\end{equation}
where $x\in\re^n$ is the state, $u\in\mathcal{L}_\infty^m$ is the input, and $f:\re^n\times\re^m\to\re^n$ is a continuous, non-Lipschitz vector field that satisfies $f(0,0)=0$. We will state some definitions and results from \cite{fxtiss_converse} that will be particularly useful for our work.
\begin{definition}
    The system \eqref{sysu} is said to be \emph{fixed-time input-to-state stable (FxT ISS)} if there exists $\beta\in\klfx$ and $\varrho\in\mathcal{K}$ such that for each $x_0\in\mathbb{R}^n$ and $u\in\mathcal{L}_\infty^m$, every solution $x(t)$ of \eqref{sysu} exists for $t\ge 0$ and satisfies
     \begin{equation*}
         |x(t)|\le \beta(|x_0|, t)+\varrho(|u|_\infty),
     \end{equation*}
     for all $t\ge 0,$
     where the settling time function $T$ of $\beta$ is continuous and uniformly bounded, with $T(0)=0$.\QEDB
\end{definition}
\vspace{0.1cm}
\begin{definition}
    A $\mathcal{C}^1$ function $V:\re^n\to\re_{\ge 0}$ is said to be a FxT ISS Lyapunov function for \eqref{sysu} if it satisfies item 1 of Definition \ref{def_fxts_lf} and there exists $\chi\in\ck_\infty$ and $\Psi\in\ckf$ such that the following holds
    \begin{equation}\label{fxtiss_imp}
        V(x)\ge \chi(|u|) \Rightarrow \nabla V(x)^\top f(x,u)\le -\Psi(V(x)),
    \end{equation}
    for all $x\in\re^n$ and $u\in\re^m$.\QEDB
\end{definition}
\vspace{0.1cm}
It is shown in \cite{fxtiss_converse} that if \eqref{sysu} admits a FxT ISS Lyapunov function, it is FxT ISS. Moreover, it  directly follows from the definition that if \eqref{sysu} is FxT ISS, then the origin of system \eqref{sysu} with $u(t)\equiv 0$ is FxTS.

\section{Main Results}\label{sec_main}
\subsection{Analysis}
We consider interconnections of the following form
\begin{subequations}\label{sysnu}
    \begin{align}
        \dot{x}_1&=f_1(x_1,x_2)\label{sysnu1}\\
        \dot{x}_2&=f_2(x_1,x_2),\label{sysnu2}
    \end{align}
\end{subequations}
where $x_i\in\re^{n_i}$ are the states and $f_i:\re^{n_1}\times\re^{n_2}\to\re^{n_i}$ are continuous vector fields with $f_1(0,0)=f_2(0,0)=0$, see Figure \ref{block}. Our primary goal is to derive lower order conditions on the subsystems \eqref{sysnu1} and \eqref{sysnu2} that allow us to establish FxTS of the system \eqref{sysnu}. To address this, we propose a Lyapunov-based small gain approach and make the following assumption on \eqref{sysnu}:
\begin{assumption}\label{assump_sysnu}
    Consider the system \eqref{sysnu} and, for each $i=1,2$, with $j=3-i$, there exists $\mathcal{C}^1$ functions $V_i:\re^{n_i}\to\re_+$ satisfying the following conditions:
    \begin{enumerate}
        \item There exists $\underline{\alpha}_i, \overline{\alpha}_i\in\ck_\infty$ such that
        \begin{equation}\label{assump_sandw}
            \underline{\alpha}_i(|x_i|)\le V(x_i)\le \overline{\alpha}_i(|x_i|),
        \end{equation}
        for all $x_i\in\re^{n_i}$.
        \item There exists $\gamma_{i}\in\ck_\infty$ and $\Psi_i\in\ckf$ such that the following holds
    \begin{align}\label{fxtiss_imp_assump}
        &V_i(x_i)\ge \gamma_{i}(V_{j}(x_{j}))\notag \\&\Rightarrow \nabla V_i(x_i)^\top f_i(x_1, x_2)\le -\Psi_i(V_i(x_i)),
    \end{align}
    for all $x\in \re^{n_1+n_2}$. \QEDB
    \end{enumerate}
\end{assumption}
\vspace{0.1cm}
In other words, we simply require that each $x_i$ subsystem is FxT ISS with respect to ``input" $x_j$, and that there exists a FxT ISS Lyapunov function to establish this property. Note that although \eqref{fxtiss_imp_assump} is not exactly in the form \eqref{fxtiss_imp}, it can be placed into the form \eqref{fxtiss_imp} by leveraging \eqref{assump_sandw}. We are now ready to state the first main result of the paper:

\vspace{0.1cm}
\begin{thm}\label{thm_sgt}
    Let system \eqref{sysnu} satisfy Assumption \ref{assump_sysnu} and suppose, for each $i=1,2$, there exists $\hat{\gamma}_{i}\in\ckp\cup\cki$ such that $\hat{\gamma}_{i}(s)>\gamma_{i}(s)$ and the following small-gain condition holds:
    \begin{equation}
    \hat{\gamma}_{1}\circ\hat{\gamma}_{2}(s)<s,
    \end{equation}
    for all $s>0$.
    Then, the origin of the system \eqref{sysnu} is FxTS.\QEDB
\end{thm}
\begin{figure}[t!]
  \centering \includegraphics[width=0.35\textwidth]{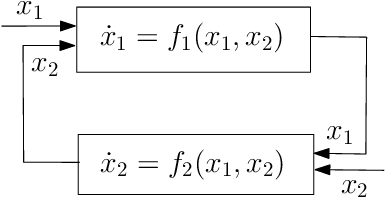}
    \caption{A block diagram depicting the interconnection \eqref{sysnu}.} \label{block}
    \vspace{-0.4cm}
\end{figure}

\vspace{0.1cm}
\begin{proof}
    Let $\sigma_\lambda\in\ck_\infty^+$ take the form $\sigma_\lambda(s)=s^\lambda$ for $\lambda\ge 1$, and consider the FxTS Lyapunov function candidate
    \begin{align}
        V(x)&=\max\{\sigma_\lambda\circ V_1(x_1), \sigma_\lambda\circ\hat{\gamma}_{1}\circ V_2(x_2), \sigma_\lambda\circ V_2(x_2),\notag\\&~~~ \sigma_\lambda\circ\hat{\gamma}_{2}\circ V_1(x_1)\},\label{lfmax}
    \end{align}
    where $\lambda$ is a sufficiently large constant such that $\sigma_\lambda\circ\hat{\gamma}_{2}\circ V_1(x_1)$ and $\sigma_\lambda\circ\hat{\gamma}_{1}\circ V_2(x_2)$ are FxT ISS Lyapunov functions of the \eqref{sysnu1} and \eqref{sysnu2} subsystems, respectively. As shown in Lemmas \ref{lem_k1fxt} and \ref{lem_invfxt}, such a $\lambda\ge 1$ always exists. Moreover, since $\lambda\ge 1$, it follows from Lemma \ref{lem_k1fxt} that $\sigma_\lambda\circ V_1(x_1)$ and $\sigma_\lambda\circ V_2(x_2)$ are also FxT ISS Lyapunov functions of the \eqref{sysnu1} and \eqref{sysnu2} subsystems, respectively. Then, for $i,j\in\{1,2\}$, we define
    \begin{subequations}\label{sigv}
    \begin{align}
        \sigma_{i1}(s)&:=\sigma_\lambda(s)\\
        \sigma_{i2}(s)&:=\sigma_\lambda\circ\hat{\gamma}_{3-i}(s)\\
        V_{ij}(x_i)&:=\sigma_{ij}\circ V_i(x_i).
    \end{align}
    \end{subequations}
    Let $x(0)\in\mathbb{R}^n$ and let $x(\cdot)$ be a maximal solution to system \eqref{sysnu} from $x(0)$, defined for all $t\in[0,T_\text{max})$, with $T_\text{max}\in(0,\infty]$. Each $V_{ij}(x_i)$ is a FxT ISS Lyapunov function for the $x_i$ subsystem and, hence, satisfies the following properties for each $i,j\in \{1,2\}$:
    \begin{enumerate}
        \item The inequalities
        \begin{equation*}
            \underline{\alpha}_{ij}(|x_i|)\le V_{ij}(x_i)\le \overline{\alpha}_{ij}(|x_i|),
        \end{equation*}
        hold for all $x_i\in\re^{n_i}$, where $\underline{\alpha}_{ij}=\sigma_{ij}\circ\underline{\alpha}_i$ and $\overline{\alpha}_{ij}=\sigma_{ij}\circ\overline{\alpha}_i$.
        \item There exists $\Psi_{ij}\in\ckf$ such that:
        \begin{align}\label{subiss}
            &V_{ij}(x_i)\ge \gamma_{ij}\circ V_{3-i}(x_{3-i})\notag\\&\Rightarrow D^+ V_{ij}(x_i)\le -\Psi_{ij}(V_{ij}(x_i)),
        \end{align} 
        holds for all $0\le t<T_\text{max}$ and $x\in\re^{n_1+n_2}$, where $\gamma_{ij}=\sigma_{ij}\circ\hat{\gamma}_i$.
    \end{enumerate}
    Next, define
    \begin{align*}
        \underline{\alpha}(s)&=\min_{i,j=1,2}\underline{\alpha}_{ij}\left(\frac{s}{\sqrt{2}}\right),\quad
        \overline{\alpha}(s)=\max_{i,j=1,2}\overline{\alpha}_{ij}(s).
    \end{align*}
    It is easy to verify that $V(x)=\max_{i,j=1,2} V_{ij}(x_i)$ satisfies \eqref{sandw} with $\underline{\alpha}, \overline{\alpha}\in\ck_\infty$. Now, consider some $t\in [0, T_\text{max})$ where $V(x)=\sigma_{ij}\circ V_i(x_i)$, which leads to four cases:
    \begin{enumerate}
        \item If $i=1$ and $j=1$, we have that $\sigma_{11}\circ V_1(x_1)\ge \sigma_{22}\circ V_2(x_2)$. This implies $V_1(x_1)\ge \hat{\gamma}_1\circ V_2(x_2)$, and thus $V_{11}(x_1)\ge \gamma_{11}\circ V_2(x_2)$. By \eqref{subiss}, we have $\nabla V_{11}(x_1)^\top f_1(x_1, x_2)\le -\Psi_{11}(V_{11}(x_1))$
        \item If $i=1$ and $j=2$, we have $\sigma_{12}\circ V_1(x_1)\ge \sigma_{21}\circ V_2(x_2)$, which implies $\hat{\gamma}_2\circ V_1(x_1)\ge V_2(x_2)$. By the small gain condition, we also have $V_1(x_1)\ge \hat{\gamma}_1\circ V_2(x_2)$, which implies $V_{12}(x_1)\ge \gamma_{12}\circ V_2(x_2)$. Hence, by \eqref{subiss}, it follows that $\nabla V_{12}(x_1)^\top f_1(x_1, x_2)\le -\Psi_{12}(V_{12}(x_1))$.
        \item If $i=2$ and $j=1$, we have $\sigma_{21}\circ V_2(x_2)\ge \sigma_{12}\circ V_1(x_1)$, which implies $V_2(x_2)\ge \hat{\gamma}_2\circ V_1(x_1)$, and thus $V_{21}(x_2)\ge \gamma_{21}\circ V_1(x_1)$. By \eqref{subiss}, we have $\nabla V_{21}(x_2)^\top f_2(x_1, x_2)\le -\Psi_{21}(V_{21}(x_2))$.
        \item If $i=2$ and $j=2$, we have $\sigma_{22}\circ V_2(x_2)\ge \sigma_{11}\circ V_1(x_1)$, which implies $\hat{\gamma}_1\circ V_2(x_2)\ge V_1(x_1)$. By the small gain condition, we have $V_2(x_2)\ge \hat{\gamma}_2\circ V_1(x_1)$, which implies $V_{22}(x_2)\ge \gamma_{22}\circ V_1(x_1)$. Hence, \eqref{subiss} yields $\nabla V_{22}(x_2)^\top f_2(x_1, x_2)\le -\Psi_{22}(V_{22}(x_2))$.
    \end{enumerate}
    Since $\Psi_{i,j}\in\ckf$ for each $i,j$, it follows from Lemma \ref{lem_fxtbd} that there exists some $\Psi\in\ckf$ such that \begin{equation*}
        \Psi(s)\le \min_{i,j=1,2}\Psi_{ij}(s),\quad \forall s\ge 0.
    \end{equation*}
    Now we can define $I(t)=\{(i,j) : V(x(t))=V_{ij}(x_i(t))\}$ for $t\in [0, T_\text{max})$,
    and we can take the upper right Dini derivative of $V$ along the trajectories of \eqref{sysnu} to obtain
    \begin{align*}
        D^+V(x(t))&=\max_{(i,j)\in I(t)} D^+ V_{ij}(x_i(t))\\
        &=\max_{(i,j)\in I(t)} \nabla V_{ij}(x_i(t))^\top f_i(x_1(t), x_2(t))
        \\
        &\le \max_{(i,j)\in I(t)}-\Psi_{ij}(V_{ij}(x_i(t)))\\
        &\le \max_{(i,j)\in I(t)} -\Psi(V_{ij}(x_i(t)))\\
        &=-\Psi(V(x(t))),
    \end{align*}
    where the first equality follows from \cite[Lemma 2.9]{giorgi1992dini}. The trajectories are bounded on $[0, T_\text{max})$, and hence they are defined for all $t\ge 0$. This concludes the proof.
\end{proof}
\vspace{0.1cm}

Our approach is inspired by the ideas from \cite{jiang1996lyapunov,sgfinite}, in that we leverage the gain functions $\gamma_i$ to construct, for the interconnected system, a FxTS Lyapunov function candidate defined in the max form \eqref{lfmax}. Since the constructed FxTS Lyapunov function is possibly non-smooth, we must rely on the Dini derivative formulation in \eqref{fxts_lyapunov_bd} to establish FxTS. To ensure that the FxT ISS Lyapunov function properties are preserved, we also utilize the power-function-based scaling technique introduced in \cite{sgfinite} (this is the role of $\sigma_\lambda$). However, one important distinction to note is that since \cite{sgfinite} is only concerned with finite time stability notions, it is sufficient in that setting for the gains $\hat{\gamma}_i(\cdot)$ to only have a class $\ckp$ approximation near the origin. Since we are concerned with the global notion of \emph{fixed-time stability}, we now require that the class $\ckp$ approximation holds globally. One challenge that arises is that if $\hat{\gamma}_1,\hat{\gamma}_2\in\ckp$, the small-gain condition will never be satisfied if there is a $\hat{\gamma}_i$ containing terms with different powers. To address this, we also allow for the gains $\hat{\gamma}_i$ to be class $\cki$, and we show in the Lemma \ref{lem_invfxt} of the Appendix that such functions can also be appropriately scaled to preserve the FxT ISS Lyapunov function property.

\vspace{0.1cm}
\begin{remark}
    While most functions in $\cki$ do not have a closed-form expression, this is not particularly problematic for our work. Indeed, if $\hat{\gamma}_i\in\cki$, then $\hat{\gamma}_i^{-1}\in\ckp$, so we can use the following equivalence 
    \begin{equation*}
        V_i(x_i)\ge \hat{\gamma}_i(V_j(x_j))\Longleftrightarrow \hat{\gamma}_i^{-1}(V_i(x_i))\ge V_j(x_j),
    \end{equation*}
    to verify \eqref{fxtiss_imp_assump} in a simplified manner. Moreover, if $\hat{\gamma}_1\in\ckp$ and $\hat{\gamma}_2\in\cki$, we can leverage the following fact
    \begin{equation}\label{sgeq}
        \hat{\gamma}_1\circ\hat{\gamma}_2(s)<s\Longleftrightarrow \hat{\gamma}_1(s)<\hat{\gamma}_2^{-1}(s),
    \end{equation}
    to verify the small-gain condition. Since $\hat{\gamma}_2^{-1}\in\ckp$, the condition $\hat{\gamma}_1(s)<\hat{\gamma}_2^{-1}(s)$ can be checked using straightforward methods, such as by applying Lemma \ref{lem_sandw}. This is further detailed in the following example.\QEDB
\end{remark}
\vspace{0.1cm}
\subsection{Second-Order example}
Consider the dynamics
\begin{subequations}\label{ex0}
    \begin{align}
        \dot{x}&=f(x)+\varepsilon_1\left(\sg{y}^{p}+\sg{y}^{q}\right)\\
        \dot{y}&=g(y)-\varepsilon_2\sg{x}^r,\label{exsuby0}
    \end{align}
\end{subequations}
where $p, q, r>0$ and $\sg{\cdot}^a=|\cdot|^a\text{sgn}(\cdot)$. We assume that $f(\cdot)$ and $g(\cdot)$ are continuous vector fields such that the origins of the systems $\dot{x}=f(x)$ and $\dot{y}=g(y)$ are FxTS. We also assume that they admit quadratic FxTS Lyapunov functions $V(x)=x^2$ and $W(y)=y^2$, respectively. We will show that the origin for the system \eqref{ex0} is FxTS, as long as $\varepsilon_1, \varepsilon_2>0$ are sufficiently small and $p,q,r$ satisfy some conditions specified below.
\begin{thm}\label{thm_ex}
    Consider the system \eqref{ex0} with $p,q,r>0$ and $f:\re\to\re, g:\re\to\re$ are continuous functions for which there exists $m_1, m_2>0, \alpha_1, \alpha_2\in (0,1)$, and $\beta_1, \beta_2>1$ such that the following inequalities hold:
    \begin{subequations}
    \begin{align*}
        2xf(x)&\le -2m_1 |x|^{2\alpha_1}-2m_1 |x|^{2\beta_1}\\2yg(y)&\le -2m_2 |y|^{2\alpha_2}-2m_2 |y|^{2\beta_2},
    \end{align*}
    \end{subequations}
    for all $x,y\in\re$. Moreover, let $L_r=\max(1, 2^{r-1})$ and suppose $pr,qr\in [\alpha_2, \beta_2]$ and $\varepsilon_1, \varepsilon_2$ satisfy
    \begin{equation}\label{epbd}
    0<\varepsilon_1^r\varepsilon_2<\frac{m_1^r m_2}{2^r L_r^\frac12}.
    \end{equation}
    Then, the origin of \eqref{ex0} is FxTS.\QEDB
\end{thm}
\vspace{0.1cm}
\begin{proof}
    Taking $V(x)=x^2$ and $W(y)=y^2$, we have:
    \begin{align*}
    \dot{V}&=2xf(x)+2\varepsilon_1x\sg{y}^p+2\varepsilon_1x\sg{y}^q\\
    &\le -2m_1 V^{\alpha_1}-2m_1 V^{\beta_1}+2|x|(\varepsilon_1|y|^p+\varepsilon_1|y|^q).
\end{align*}
Let $c_1\in (0, m_1)$, which yields the following implication
\begin{equation*}
    |x|>\frac{\varepsilon_1}{m_1-\frac12 c_1}(|y|^p+|y|^q)\Rightarrow \dot{V}\le -c_1 V^{\alpha_1}-c_1 V^{\beta_1}.
\end{equation*}
Thus, we have
\begin{align*}    \gamma_1(s)&:=\left(\frac{\varepsilon_1}{m_1-\frac12 c_1}\right)^2 (s^{\frac{p}{2}}+s^{\frac{q}{2}})^2\\&\le 2\left(\frac{\varepsilon_1}{m_1-\frac12 c_1}\right)^2(s^p+s^q).
\end{align*}
For the $y$ subsystem, we have
\begin{align*}
    \dot{W}&=2yg(y)-2\varepsilon_2 y\sg{x}^r\\
    &\le -m_2 W^{\alpha_2}-m_2 W^{\beta_2}+ \frac{\varepsilon_2^2}{m_2}|x|^{2r}.
\end{align*}
If we fix some $c_2\in (0, m_2)$, we obtain the following
\begin{align*}
    \dot{W}&\le -c_2 W^{\alpha_2}-c_2 W^{\beta_2}\\&~~~+\left(\frac{\varepsilon_2^2}{m_2} |x|^{2r}-(m_2-c_2)\left(W^{\alpha_2}+W^{\beta_2}\right)\right).
\end{align*}
Let $L_r=\max(1, 2^{r-1})$, and suppose $p,q,r$ satisfy $pr,qr\in [{\alpha_2}, {\beta_2}]$. We want $\varepsilon_1, \varepsilon_2$ sufficiently small such that
\begin{equation}\label{gbd}
    \gamma_2^{-1}(s):=\left(\frac{m_2(m_2-c_2)}{L_r\varepsilon_2^2}\right)^\frac{1}{r}\left(s^{\frac{\alpha_2}{r}}+s^{\frac{\beta_2}{r}}\right)> \gamma_1(s),
\end{equation}
for all $s> 0$. By Lemma \ref{lem_sandw}, it suffices to have
\begin{equation*}
    \left(\frac{m_2(m_2-c_2)}{L_r\varepsilon_2^2}\right)^\frac{1}{r}\ge \frac{4\varepsilon_1^2}{(m_1-\frac12c_1)^2},
\end{equation*}
which can be obtained whenever

\begin{equation}\label{barep}
    \varepsilon_1^r\varepsilon_2\le\frac{(m_1-\frac12 c_1)^r m_2^\frac12 (m_2-c_2)^\frac12}{2^{r}L_r^\frac12}.
\end{equation}
Now, suppose $V(x)<\gamma_2^{-1}(W(y))$, which means:
\begin{equation*}
    |x|^2<\left(\frac{m_2(m_2-c_2)}{L_r\varepsilon_2^2}\right)^\frac{1}{r}\left(W^{\frac{\alpha_2}{r}}+W^{\frac{\beta_2}{r}}\right).
\end{equation*}
Raising to the $r$ power yields:
\begin{align*}
    |x|^{2r}
    &\le \frac{m_2(m_2-c_2)}{\varepsilon_2^2}(W^{\alpha_2}+W^{\beta_2}).
\end{align*}
Thus,
\begin{equation*}
    \dot{W}\le -c_2 W^{\alpha_2}-c_2 W^{\beta_2},\quad W(y)>\gamma_2(V(x)).
\end{equation*}

Recall that $c_1$ and $c_2$ can be arbitrarily small. Hence, if $\varepsilon_1^r\varepsilon_2$ satisfies \eqref{epbd}, it follows immediately that $c_1, c_2$ can be chosen such that $\varepsilon_1^r\varepsilon_2$ also satisfies \eqref{barep}. Then \eqref{gbd} implies $\gamma_2(s)<\gamma_1^{-1}(s)$ for all $s>0$. We can pick $\varepsilon>0$  such that \eqref{sgeq} is satisfied with $\hat{\gamma}_i(s)=(1+\varepsilon)\gamma_i(s)$, allowing us to conclude that the small-gain condition is satisfied. Hence, the origin of the system \eqref{ex0} is FxTS.
\end{proof}
\vspace{0.1cm}

To verify our results numerically, we plot the trajectories of \eqref{ex0} with varying initial conditions, where the parameters are chosen to yield the following system
\begin{subequations}\label{ex0plot}
    \begin{align}        \dot{x}&=-x^\frac13-x^3+0.3\left(\sg{y}^{2}+\sg{y}^{2.5}\right)\\
        \dot{y}&=-\sg{y}^\frac12-\sg{y}^2-0.1\sg{x}^{\frac49}.
    \end{align}
\end{subequations}

\begin{figure}[t!]
  \centering \includegraphics[width=0.4\textwidth]{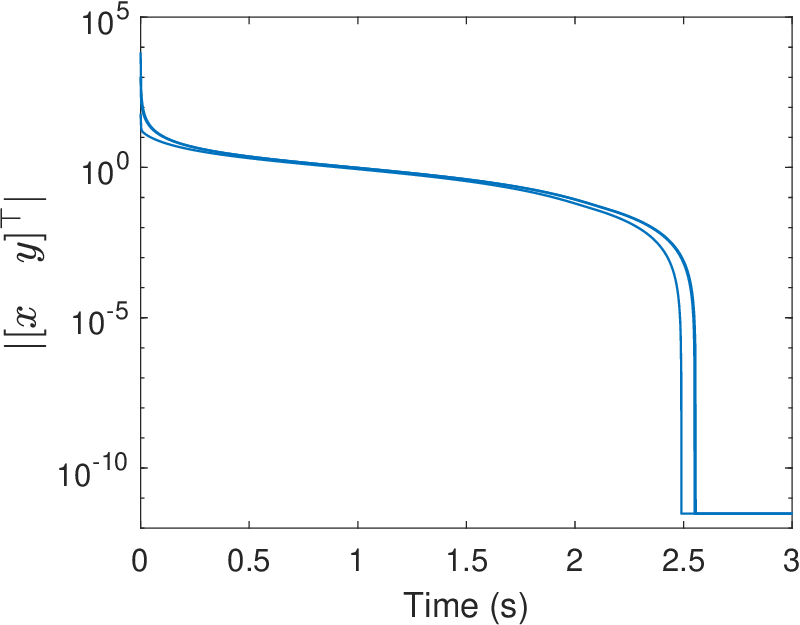}
    \caption{Trajectories of system \eqref{ex0plot} with varying $x(0),y(0)$.} \label{fxt1plot}
    \vspace{-0.4cm}
\end{figure}

It can be verified that this choice of parameters satisfy \eqref{epbd}. The trajectories are shown in Figure \ref{fxt1plot}, illustrating the FxTS property.
\begin{remark}
    One important observation that can be made is that the presence of multiple exponents of $V$ in the classical FxTS and FxT ISS Lyapunov bounds in \eqref{fxts_lyapunov_bd} and \eqref{fxtiss_imp} allows us to handle $\ckp$ gains and interconnections that contain terms with differing exponents. For standard Lyapunov-based nonlinear small-gain theory \cite{jiang1996lyapunov}, the structure of the rate of decrease of the Lyapunov function (upon satisfaction of the ISS condition) results in a limited class of admissible gain functions. Even in the case of finite-time ISS \cite{sgfinite} for interconnections, gain functions in $\ckp$ often only contain one term, the exponents of the different gain functions need to multiply to 1. Indeed, if the origins of $\dot{x}=f(x)$ and $\dot{x}=g(x)$ are only assumed to be finite-time stable, then \eqref{ex0} will not satisfy the finite-time small-gain condition from \cite{sgfinite} if $pr\neq 1$ or $qr\neq 1$. In our setting, this can be relaxed by leveraging the different exponents of $V$ and Lemma \ref{lem_sandw}. Instead of precisely requiring $pr=qr=1$, or equivalently $p=q=\frac{1}{r}$, we now only require $pr, qr\in [\alpha_2, \beta_2]$, which is a neighborhood of 1. This opens the door to generalizations of \eqref{ex0} that contain arbitrarily many cross terms between the subsystems, which will be of interest for future work. \QEDB
\end{remark}

\section{Fixed-Time Feedback Optimization without Timescale Separation}
\label{sec_fixedtime}
In this section, we apply our results to study an open problem of practical interest: feedback optimization of dynamical systems with fixed-time convergence without timescale separation. Standard feedback optimization schemes require inducing a timescale separation between the plant and controller \cite{9075378,8673636,bianchin2022online}, as it enables stability analysis with tools from singular perturbation theory. This was also the case recently in \cite{11239427}, where the authors study multi-timescale fixed-time feedback optimization via composite Lyapunov functions. In contrast to this setting, we will show that, under suitable assumptions on the plant and controller, approaching this problem with small-gain theory eliminates the need for a time-scale separation. We consider plants of the form
\begin{equation}\label{exsys}
    \dot{x}=f(x,u),
\end{equation}
where $x\in\re^n$ is the plant state and $u\in\re^m$ is a measurable and bounded control input. We will make the following assumption on the plant \eqref{exsys}:
\begin{assumption}\label{assump_ex}
    For system \eqref{exsys}, there exists $P\in\re^{n\times m}$ such that $f(Pu,u)=0$ for all $u\in\re^m$. Moreover, there exists $\ell,a,b>0, p\in (0,1), q>1$ such that
    \begin{align}\label{ex_fxtiss}
        &|x-Pu^*|\ge\ell |u-u^*|\notag\\&\Rightarrow 2(x-Pu^*)^\top f(x,u)\le -a|x-Pu^*|^{2p}-b|x-Pu^*|^{2q},
    \end{align}
    for all $x\in\re^n$, $u\in\re^m$, and $u^*\in\re^m$.\QEDB
\end{assumption}
\vspace{0.1cm}
\begin{remark}
    The mapping $Pu$ is referred to as the \emph{steady-state mapping} for the dynamics \eqref{exsys}. In works that study feedback optimization using singular perturbation theory, it is typically the case that the steady state mapping is assumed to exhibit appropriate stability properties uniformly in $u$ (see, e.g \cite{11239427, 9075378, bianchin2022online}). Here, we impose a very similar condition but in the FxT ISS framework. Essentially, we are requiring that for each fixed reference input $u^*$, the deviation $x-Pu^*$ is FxT ISS with respect to the deviation on the ``input" $u-u^*$, which is established via the quadratic FxT ISS Lyapunov function $|x-Pu^*|^2$. We require that this property holds uniformly in $x, u$, and $u^*$. One class of systems satisfying Assumption \ref{assump_ex} are those of the following form
    \begin{align}\label{accplant}
    \dot{x}=-\frac{A_1(x-P u)}{|x-P u|^{\tilde{p}}}- \frac{A_2(x-P u)}{|x-P u|^{\tilde{q}}},
\end{align}
where $A_1, A_2\in\re^{n\times n}$ with $A_1, A_2\succ 0$, $\tilde{p}\in (0,1)$, and $\tilde{q}<0$. The system \eqref{accplant} is a significant generalization of the plants considered in \cite{10644358}, where the authors used $m=n$ and $A_1=A_2=P=I_n$.
For more insight into system \eqref{accplant}, consider the following scalar system
\begin{equation}
    \dot{x}=wx+cv
\end{equation}
where $w,c\in\re\setminus\{0\}$ are known parameters and $v$ is a control input. Then, the following state feedback law:
\begin{equation*}
    v=-\frac{w}{c}u-\frac{2|w|}{c}\left(\frac{x-u}{|x-u|^{\tilde{p}}}+ \frac{x-u}{|x-u|^{\tilde{q}}}\right),
\end{equation*}
where $u\in\re$ is an external input, renders the system into a closed form that satisfies Assumption \ref{assump_ex}. Extensions to multi-variable settings can be done using results from the FxT control literature \cite{6104367, 10018220, polyakov2024fixed}.
\QEDB
\end{remark}
\vspace{0.1cm}

Our primary goal is to design a dynamic feedback law on $u$ to stabilize \eqref{exsys} while also optimizing a cost function $\phi(x,u)$. For this work, we will be primarily concerned with quadratic cost functions of the following form
\begin{equation}\label{quad_cost}
    \phi(x,u)=\frac12 x^\top Q_1 x+b_1^\top x+\frac12 u^\top Q_2 u+b_2^\top u,
\end{equation}
where $Q_1\in\re^{n\times n}, Q_2\in\re^{m\times m}, b_1\in\re^n$, and $b_2\in\re^m$ with $Q_1, Q_2\succ 0$. Since the goal is to optimize \eqref{quad_cost} while stabilizing \eqref{exsys}, we can substitute the steady state mapping $x=Pu$ into \eqref{quad_cost} to arrive at the following optimization problem
    \begin{align}\label{ex_opt}
        \min_u \Phi(u),
    \end{align}
where $\Phi(u)=\phi(Pu, u)$, with $\Phi(u)$ given by
\begin{equation*}
    \Phi(u):=\frac12 u^\top (P^\top Q_1 P+Q_2)u+(P^\top b_1+b_2)^\top u.
\end{equation*}
To solve \eqref{ex_opt}, we propose a control law based on a fixed-time gradient flow on $u$, i.e $\dot{u}=-\frac{\nabla\Phi(u)}{|\nabla\Phi(u)|^{\xi_1}}-\frac{\nabla\Phi(u)}{|\nabla\Phi(u)|^{\xi_2}}$,
where $\xi_1\in (0,1)$ and $\xi_2<0$. Following the ideas from \cite{11239427, 9075378, bianchin2022online}, after evaluating $\nabla\Phi(u)$, we can substitute $x$ back into its steady state value to obtain the following feedback controller on $u$:
\begin{equation}\label{ex_fdbk}
    \dot{u}=-\frac{\cf(x,u)}{|\cf(x,u)|^{\xi_1}}-\frac{\cf(x,u)}{|\cf(x,u)|^{\xi_2}},
\end{equation}
where $\cf(x,u)$ is given by
\begin{equation*}
    \cf(x,u)=P^\top Q_1 x+Q_2 u+P^\top b_1+b_2.
\end{equation*}
We are now ready to state the second main result of this paper:

\vspace{0.1cm}
\begin{thm}\label{thm_fbk}
    Consider system \eqref{exsys} interconnected with \eqref{ex_fdbk}, where $\xi_1\in(0,1), \xi_2<0$, and system \eqref{exsys} satisfies Assumption \ref{assump_ex}. Moreover, suppose the following condition holds
    \begin{equation}\label{ex_sgc}
        \ell\frac{\smax(P)\lmax(Q_1)}{\lmin(Q_2)}<1.
    \end{equation}
    Then, the point $(Pu^*, u^*)$ is rendered FxTS, where $u^*$ is the solution of \eqref{ex_opt} and given by
    \begin{equation*}
        u^*=-(P^\top Q_1 P+Q_2)^{-1}(P^\top b_1+b_2).
    \end{equation*}\QEDB
\end{thm}
\begin{proof}
Denote $V(\tx)=|\tx|^2$. We use the change of variables $\tilde{x}=x-Pu^*$ and $\tilde{u}=u-u^*$ to arrive at the following system
    \begin{subequations}
        \begin{align}
            \dot{\tilde{x}}&=f(\tilde{x}+Pu^*, \tilde{u}+u^*)\label{extx}\\
            \dot{\tilde{u}}&=-\frac{\tcf(\tx,\tu)}{|\tcf(\tx,\tu)|^{\xi_1}}-\frac{\tcf(\tx,\tu)}{|\tcf(\tx,\tu)|^{\xi_2}},\label{extu}
        \end{align}
    \end{subequations}
    where $\tcf=P^\top Q_1 \tx+Q_2\tu$.
    By \eqref{ex_sgc}, we can pick $\varepsilon\in (0, \lmin(Q_2))$ such that 
\begin{equation}\label{ex_sgcond}
    \ell\frac{\smax(P)\lmax(Q_1)}{\lmin(Q_2)-\varepsilon}<1.
\end{equation}
    Consider the FxT ISS Lyapunov function $W(\tu)=|\tu|^2$ for \eqref{extu}. It follows by direct computation that if $|\tu|\ge \frac{\smax(P)\lmax(Q_1)}{\lmin(Q_2)-\varepsilon}|\tx|$, then the following inequalities hold:
    \begin{subequations}\label{expfeq}
        \begin{align}
            \tu^\top P^\top Q_1 \tx+\tu^\top Q_2 \tu&\ge \varepsilon |\tu|^2\\
            |P^\top Q_1 \tx+Q_2 \tu|^{\xi_1}
    &\le (2\lmax(Q_2))^{\xi_1}|\tu|^{\xi_1}\\
    |P^\top Q_1 \tx+Q_2 \tu|^{-\xi_2}
    &\ge \varepsilon^{-\xi_2}|\tu|^{-\xi_2}.
        \end{align}
    \end{subequations}
    We can then leverage the inequalities \eqref{expfeq} with the FxT ISS Lyapunov function $W(\tu)=|\tu|^2$ to obtain that, for $W(\tu)\ge \left(\frac{\smax(P)\lmax(Q_1)}{\lmin(Q_2)-\varepsilon}\right)^2V(\tx)$, the following holds along the trajectories of \eqref{extu}:
    \begin{align*}
        \dot{W}&\le -\frac{2\varepsilon}{(2\lmax(Q_2))^{\xi_1}}W^{1-\frac12\xi_1}-2\varepsilon^{1-\xi_2}W^{1-\frac12\xi_2}.
    \end{align*}
    Moreover, by Assumption 2, we know that for $V(\tx)\ge \ell^2 W(\tu)$, the following holds along the trajectories of \eqref{extx}
    \begin{align*}
        \dot{V}&=2\tx^\top f(\tilde{x}+Pu^*, \tilde{u}+u^*)\le -a V^p-bV^q.
    \end{align*}
    By \eqref{ex_sgcond} and Theorem \ref{thm_sgt}, we conclude that the origin for the $(\tx, \tu)$ system is FxTS, which establishes the result.
\end{proof}

\begin{figure}[t!]
  \centering \includegraphics[width=0.4\textwidth]{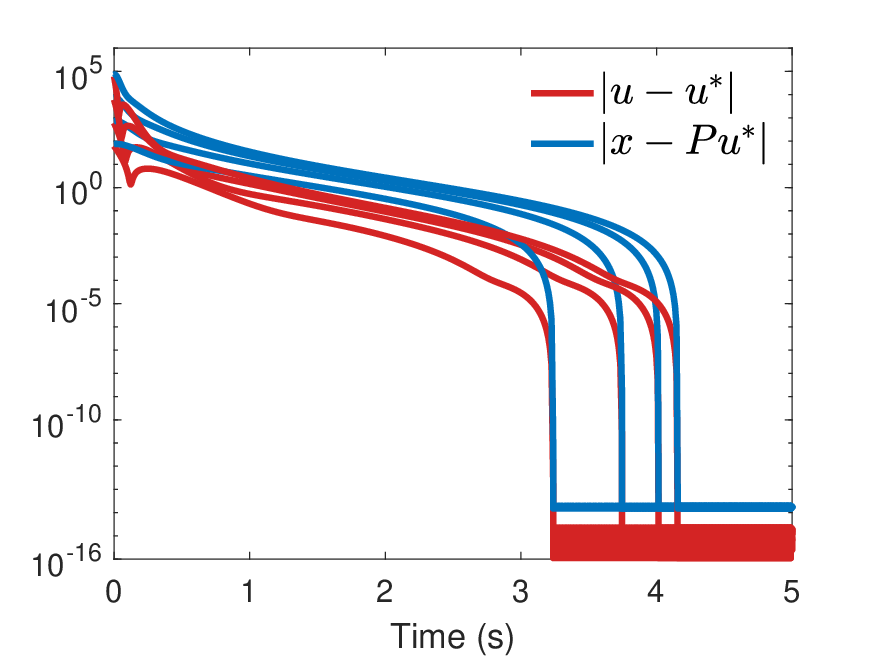}
    \caption{Trajectories of system \eqref{accplant} with \eqref{ex_fdbk} and varying initial conditions, where the parameters are given by \eqref{accplant_param} and \eqref{quadcost_param}.} \label{fxtsgplot}
    \vspace{-0.4cm}
\end{figure}

To illustrate Theorem \ref{thm_fbk} numerically, we consider system \eqref{accplant} interconnected with the dynamic controller \eqref{ex_fdbk}, where $x,u\in\re^2$ and
\begin{align}\label{accplant_param}
    A_1=\begin{bmatrix}
        \frac32 & \frac15\\ \frac15 & \frac65
    \end{bmatrix},\quad A_2=\begin{bmatrix}
        \frac65 & \frac{3}{10} \\ \frac{3}{10}&\frac85
    \end{bmatrix},\quad P=\begin{bmatrix}
        1 & \frac23\\ \frac43 &\frac13
    \end{bmatrix},
\end{align} and the cost function parameters of \eqref{quad_cost} are given by
\begin{equation}\label{quadcost_param}
    Q_1=\begin{bmatrix}
        \frac12 & \frac14\\ \frac14& \frac15
    \end{bmatrix},\quad Q_2=\begin{bmatrix}
        4 &\frac25 \\ \frac25 & 4
    \end{bmatrix},\quad b_1=\begin{bmatrix}
        3\\ 2
    \end{bmatrix},\quad b_2=\begin{bmatrix}
        1\\ 1
    \end{bmatrix}.
\end{equation}
It is straightforward to verify that \eqref{accplant} with \eqref{accplant_param} satisfies Assumption \ref{assump_ex} with $\ell=3.05$, and that the given parameters also satisfy \eqref{ex_sgc}. The trajectories are simulated with varying initial conditions in Figure \ref{fxtsgplot}.
\section{Conclusion}
\label{sec_conclusions}
We introduced a Lyapunov-based small-gain theorem for establishing fixed-time stability (FxTS) in interconnected dynamical systems. To the best of our knowledge, this is the first result of its kind in the literature on fixed-time stability. The theoretical developments are demonstrated through two illustrative examples. The first example, a second-order interconnection, provides a clear paradigm for applying the proposed tools to study interconnections of FxTS vector fields with homogeneous coupling terms. In the second example, we apply our results to analyze fixed-time feedback optimization in dynamical systems without requiring time-scale separation between the plant and the controller. Future work includes extending the result to large-scale networked and hybrid systems. Extensions to FxT ISS and applications to games can be found in \cite{tang2025lyapunov}. 

\bibliographystyle{IEEEtran}
\bibliography{Biblio.bib}
\section{Appendix}
\begin{lemma}\label{lem_fxtbd}
    If $\alpha_1,...,\alpha_N\in\ckf$, there exists $\alpha\in\ckf$ such that $\min_k\alpha_k(s)\ge \alpha(s)$ for all $s\ge 0$.\QEDB
\end{lemma}
\begin{proof}
    Let $\alpha_i(s)=a_i s^{p_i}+b_i s^{q_i}$, where $p_i\in (0,1)$ and $q_i>1$ for $i=1,...,N$. Denote $m_i=\min(a_i, b_i)$ and $m=\min_k m_k$. By Lemma \ref{lem_sandw}, we have that $\alpha_i(s)\ge m s^r$ for all $s\ge0$, $i\in [N]$, and $r\in\cap_{j=1}^N [p_j, q_j]=[\max_k p_k, \min_k q_k]$. Fix $p\in [\max_k p_k, 1)$ and $q\in (1, \min_k q_k]$, then
    \begin{equation*}
        \min_k\alpha_k(s)\ge \frac{m}{2}(s^p+s^q)\in\ckf,
    \end{equation*}
    which establishes the result. 
\end{proof}
\vspace{0.1cm}
\begin{lemma}\label{lem_k1scale}
    Let $\Psi\in\ckf$ and $\sigma(s)=\sum_{i=1}^n c_i s^{r_i}$, where $1\le r_1< r_2< ..< r_n$ and $c_i>0$ for $i\in [n]$. Then, there exists $\tilde{\Psi}\in\ckf$ such that $\tilde{\Psi}(\sigma(s))\le\Psi(s)\sigma'(s)$ for all $s>0$.\QEDB
\end{lemma}
\begin{proof}
    Let $\Psi(s)=a s^{p}+bs^q$, where $p\in (0,1)$ and $q>1$. It then follows that $\sigma'(s)=\sum_{i=1}^n r_i c_i s^{r_i-1}$, which yields $\Psi(s)\sigma'(s)\ge ms^{p+r_1-1}+ms^{q+r_n-1},$
    where $m=\min(ar_1 c_1, br_n c_n)$. We then postulate $\tilde{\Psi}\in\ckf$ of the form $\tilde{\Psi}(s)=\tilde{m}s^{\tilde{p}}+\tilde{m}s^{\tilde{q}}$, where $\tilde{p}=\frac{p+r_1-1}{r_1}\in (0,1)$ and $\tilde{q}=\frac{q+r_n-1}{r_n}>1$. Indeed, computations yield:
    \begin{align}
        \tilde{\Psi}(\sigma(s))&\le\tilde{m}\left(\sum_{i=1}^n c_i^{\tilde{p}}+n^{\tilde{q}-1}\sum_{j=1}^n c_j^{\tilde{q}}\right)(s^{p+r_1-1}+s^{q+r_n-1}),\notag
    \end{align}
    where the inequality follows from Lemma \ref{lem_sandw} and the fact that $r_i\tilde{p}, r_i \tilde{q}\in [p+r_1-1, q+r_n-1]$ for all $i\in [n]$. We can then pick $\tilde{m}$ sufficiently small to obtain the result.
\end{proof}
\vspace{0.1cm}
\begin{lemma}\label{lem_invscale}
    Suppose $\gamma\in\cki$, and let $\Psi\in\ckf$. Then there exists $\lambda\ge 1$ and $\tilde{\Psi}\in\ckf$ such that $\varrho(s)\ge \tilde{\Psi}(\gamma^\lambda(s))$ for all $s> 0$, where $\varrho(s)=\lambda\gamma^{\lambda-1}(s)\gamma'(s)\Psi(s)$.\QEDB
\end{lemma}
\begin{proof}
    Denote $\Psi(s)=a s^{p}+b s^{q}$, where $p\in (0,1)$ and $q>1$ and let $\gamma^{-1}(s)=\sum_{i=1}^n c_i s^{r_i}$, where $0<r_1<r_2<...<r_n$ and $c_i>0$ for all $i\in [n]$. 
    Note that $\gamma^{-1}(s)\ge c_i s^{r_i}$, which implies $\gamma(s)\le \tilde{c}_i s^\frac{1}{r_i}$ for all $i\in [n]$ and $s\ge 0$, where $\tilde{c}_i=c_i^{-\frac{1}{r_i}}$. From this, we obtain that $\Psi(s)\ge a c_i^{p}\gamma^{p r_i}(s)+b c_i^{q}\gamma^{q r_i}(s).$ We define the constant $R:=\sum_{i=1}^n r_i c_i$, and we first consider the case of $\gamma(s)\in (0,1)$, which implies $\sum_{i=1}^n r_i c_i \gamma^{r_i-1}(s)\le R\gamma^{r_1-1}(s)$. From the inverse function theorem, we obtain:
    \begin{align*}
        \varrho(s)&=\frac{\lambda\gamma^{\lambda-1}(s)\Psi(s)}{\sum_{i=1}^n r_i c_i \gamma^{r_i-1}(s)}\\
        &\ge \tilde{m}_i\gamma^{\lambda+p r_i-r_1}(s)+\tilde{m}_i\gamma^{\lambda+q r_i-r_1}(s),
    \end{align*}
    where $\tilde{m}_i=\frac{\lambda}{R}\min\left(ac_i^{p}, b c_i^{q}\right)$. Choosing $i=1$ yields $\varrho(s)\ge \tilde{\Psi}_1(\gamma^\lambda(s))$,
where $\tilde{\Psi}_i=\tilde{m}_i s^{1-\frac{(1-p)r_i}{\lambda}}+\tilde{m}_i s^{1+\frac{(q-1)r_i}{\lambda}}$. Picking $\lambda>(1-p)r_1$ guarantees $\tilde{\Psi}_1\in\ckf$. Now, we consider the case of $\gamma(s)\ge 1$, which implies $\sum_{i=1}^n r_i c_i \gamma^{r_i-1}(s)\le R\gamma^{r_n-1}(s)$. We can then perform similar calculations done above to obtain:
\begin{equation*}
    \varrho(s)\ge \tilde{m}_i\gamma^{\lambda+p r_i-r_n}(s)+\tilde{m}_i\gamma^{\lambda+q r_i-r_n}(s),\quad i\in [n].
\end{equation*}
If we choose $i=n$, we obtain $\varrho(s)\ge \tilde{\Psi}_n(\gamma^\lambda(s))$,
where $\tilde{\Psi}_n\in\ckf$ if $\lambda>(1-p)r_n$. Since we assume $r_1<r_n$, it suffices to have $\lambda>(1-p)r_n$. By Lemma \ref{lem_fxtbd}, we know that there exists $\tilde{\Psi}\in\ckf$ such that $\min(\tilde{\Psi}_1(s), \tilde{\Psi}_n(s))\ge \tilde{\Psi}(s)$ for all $s>0$. Putting both cases together, we obtain
\begin{equation*}
    \varrho(s)\ge \min\left(\tilde{\Psi}_1(\gamma^\lambda(s)), \tilde{\Psi}_n(\gamma^\lambda(s))\right)\ge \tilde{\Psi}(\gamma^\lambda(s)),
\end{equation*}
which establishes the result.
\end{proof}
\vspace{0.1cm}
\begin{lemma}\label{lem_k1fxt}
    Suppose $V:\re^n\to\re_+$ is a FxT ISS Lyapunov function for \eqref{sysu}. Then, if $\sigma(s)=\sum_{i=1}^n c_i s^{r_i}$, where $1\le r_1< r_2< ..< r_n$ and $c_i>0$ for $i\in [n]$, it follows that $\tilde{V}(x)=\sigma(V(x))$ is a FxT ISS Lyapunov function of \eqref{sysu}.\QEDB
\end{lemma}
\begin{proof}
    First, it is easy to verify that $\sigma(\underline{\alpha}(|x|))\le \tilde{V}(x)\le \sigma(\overline{\alpha}(|x|)).$ Let $\tilde{\chi}=\sigma\circ \chi$ and suppose $\tilde{V}(x)\ge \tilde{\chi}(|u|)$, which implies $\nabla V(x)^\top f(x,u)\le -\Psi(V(x))$. Then we obtain:
    \begin{align*}
        \nabla\tilde{V}(x)^\top f(x,u)&= \sigma'(V(x)) \nabla V(x)^\top f(x,u)\\
        &\le -\Psi(V(x))\sigma'(V(x)),
    \end{align*}
    where $\Psi\in\ckf$ and $\sigma'$ denotes the derivative of $\sigma$. By Lemma \ref{lem_k1scale}, we can find some $\tilde{\Psi}\in\ckf$ such that $\tilde{\Psi}(\sigma(s))\le \Psi(s)\sigma'(s)$ for all $s>0$. Then we obtain
    \begin{equation*}
        \nabla\tilde{V}(x)^\top f(x,u)\le -\tilde{\Psi}(\sigma(V(x)))=-\tilde{\Psi}(\tilde{V}(x)),
    \end{equation*}
    as desired.
\end{proof}
\vspace{0.1cm}
\begin{lemma}\label{lem_invfxt}
     Suppose $V:\re^n\to\re_+$ is a FxT ISS Lyapunov function for \eqref{sysu}, and let $\gamma\in\cki$. Then, there exists $\lambda\ge 1$ such that
     $\tilde{V}(x)=\gamma^\lambda(V(x))$
     is a FxT ISS Lyapunov function for \eqref{sysu}.\QEDB
\end{lemma}
\begin{proof}
    Let $\gamma^{-1}(s)=\sum_{i=1}^n c_i s^{r_i}$, where $0<r_1<r_2<...<r_n$ and $c_i>0$ for all $i\in [n]$. Since $\gamma^{-1}$ is $\cc^1$ with a nonzero derivative on $(0, \infty)$, it follows from the inverse function theorem that $\gamma^\lambda$ is also $\cc^1$ on $(0, \infty)$, and its derivative is given by 
    \begin{equation*}
        (\gamma^\lambda)'(s)=\frac{\lambda\gamma^{\lambda-1}(s)}{(\gamma^{-1})'(\gamma(s))}=\frac{\lambda\gamma^{\lambda-1}(s)}{\sum_{i=1}^n r_i c_i \gamma^{r_i-1}(s)},\quad s>0.
    \end{equation*}
    Moreover, note that $\gamma(s)\le \tilde{c}_i s^\frac{1}{r_i}$ for all $i\in[n]$ and $s\ge 0$, which implies
    \begin{align*}
        0\le\lim_{h\to 0^+}\frac{\gamma^\lambda(h)-\gamma^\lambda(0)}{h}\le
        \lim_{h\to 0^+}\frac{\tilde{c}^\lambda_i h^\frac{\lambda}{r_i}}{h}.
    \end{align*}
    By choosing $i=1$, we observe that $\lambda>r_1$ results in $(\gamma^\lambda)'(0)=0$, which implies that $\gamma^\lambda$ is $\mathcal{C}^1$ on $[0, \infty)$. 
    Suppose $\tilde{V}(x)\ge \tilde{\chi}(|u|)$, where $\tilde{\chi}=\gamma^\lambda\circ\chi$. This implies $V(x)\ge \chi(|u|)$, and thus $\nabla V(x)^\top f(x,u)\le -\Psi(V(x))$ for some $\Psi\in\ckf$. Similar computations as done above yield:
    \begin{align*}
        \nabla\tilde{V}(x)^\top f(x,u)
        &\le -\varrho(V(x)),
    \end{align*}
    where 
    \begin{equation*}
        \varrho(s)=(\gamma^\lambda)'(s)\Psi(s)=\frac{\lambda\gamma^{\lambda-1}(s)\Psi(s)}{\sum_{i=1}^n r_i c_i \gamma^{r_i-1}(s)}.
    \end{equation*}
    We would like to find $\tilde{\Psi}\in\ckf$ such that $\varrho(s)\ge \tilde{\Psi}(\gamma^\lambda(s))$ for all $s>0$. By Lemma \ref{lem_invscale}, we can pick $\lambda$ sufficiently large such that this is possible. Thus, we have
    \begin{equation*}
        \nabla\tilde{V}(x)^\top f(x,u)\le -\tilde{\Psi}(\gamma^\lambda(V(x(t))))=-\tilde{\Psi}(\tilde{V}(x)),
    \end{equation*}
    as desired.
\end{proof}
\end{document}